\documentclass{theoretics}

\title{3SUM in Preprocessed Universes: Faster and Simpler}
\ThCSauthor[Delhi]{Shashwat Kasliwal}{shashwat.iitd.math@gmail.com}[]
\ThCSauthor[Bocconi]{Adam Polak}{adam.polak@unibocconi.it}[0000-0003-4925-774X]
\ThCSauthor[Delhi]{Pratyush Sharma}{pratyushkbs@gmail.com}[]
\ThCSaffil[Delhi]{IIT Delhi, Delhi, India}
\ThCSaffil[Bocconi]{Bocconi University, Milano, Italy}
\ThCSthanks{Invited article from SOSA 2025 \cite{KasliwalPS25}. 
  Work supported by European Research Council (ERC) under the EU’s Horizon 2020 research and innovation programme (Grant agreement No. 101019547).}
\ThCSshortnames{S.\ Kasliwal, A.\ Polak, P.\ Sharma}  
\ThCSshorttitle{3SUM in Preprocessed Universes: Faster and Simpler}
\ThCSyear{2025}
\ThCSarticlenum{24}
\ThCSreceived{Sep 6, 2024}
\ThCSaccepted{Sep 11, 2025}
\ThCSpublished{Oct 16, 2025}
\ThCSdoicreatedtrue
\ThCSkeywords{Fine-grained complexity, 3SUM, data structures, preprocessing}

\usepackage{booktabs}  
\usepackage{braket}  
\usepackage[utf8]{inputenc} 
\usepackage[T1]{fontenc}   
\usepackage{url}           
\usepackage{mathtools}

\newcommand{\Ot}{\widetilde{O}}
\newcommand{\E}{\mathbb{E}}

\makeatletter
\newcommand{\tpmod}[1]{\mkern 8mu({\operator@font mod}\mkern 6mu#1)}
\makeatother

\begin{document}

\maketitle

\begin{abstract}
We revisit the 3SUM problem in the \emph{preprocessed universes} setting. We present an algorithm that, given three sets $A$, $B$, $C$ of $n$ integers, preprocesses them in quadratic time, so that given any subsets $A' \subseteq A$, $B' \subseteq B$, $C' \subseteq C$, it can decide if there exist $a \in A'$, $b \in B'$, $c \in C'$ with $a+b=c$ in time $O(n^{1.5} \log n)$.

In contrast to both the first subquadratic $\Ot(n^{13/7})$-time algorithm by Chan and Lewenstein (STOC 2015) and the current fastest $\Ot(n^{11/6})$-time algorithm by Chan, Vassilevska Williams, and Xu (STOC 2023), which are based on the Balog--Szemer\'edi--Gowers theorem from additive combinatorics, our algorithm uses only standard 3SUM-related techniques, namely FFT and linear hashing modulo a prime. It is therefore not only faster but also simpler.

Just as the two previous algorithms, ours not only decides if there is a single 3SUM solution but it actually determines for each $c \in C'$ if there is a solution containing it. We also modify the algorithm to still work in the scenario where the set $C$ is unknown at the time of preprocessing. Finally, even though the simplest version of our algorithm is randomized, we show how to make it deterministic losing only polylogarithmic factors in the running time.
\end{abstract}
 
\section{Introduction}

In the 3SUM problem we are given three sets $A$, $B$, $C$ of $n$ integers, and we have to decide if there exists a triple $(a, b, c) \in A \times B \times C$ such that $a+b = c$. It is one of the three central problems in find-grained complexity, alongside APSP and SAT~\cite{VassilevskaWilliams18}. A standard textbook approach solves 3SUM in quadratic time, and the fastest known algorithm~\cite{BaranDP08} offers merely a factor of $\frac{\log^2 n}{(\log \log n)^2}$ improvement. It is conjectured that there is no algorithm for 3SUM running in time $O(n^{2-\varepsilon})$, for any $\varepsilon > 0$, and conditional lower bounds for numerous problems have been shown under this assumption (see, e.g.,~\cite{VassilevskaWilliams18,FischerK024}).

In order to better understand the complexity of 3SUM, it is instructive to study variants of 3SUM that do admit strongly subquadratic algorithms. For instance, a~folklore algorithm based on the Fast Fourier Transform (FFT) works in time $O(n + U \log U)$ in the case where $A, B, C \subseteq [U] := \{0, 1, \ldots, U-1\}$, i.e., this algorithm is subquadratic when $U = O(n^{2-\varepsilon})$, for $\varepsilon > 0$.

In this paper we focus on another variant, called \emph{3SUM in preprocessed universes}, first proposed by Bansal and Williams~\cite{BansalW12}. In this variant, we are first given three sets $A$, $B$, $C$, and we can preprocess them in polynomial (ideally, near-quadratic) time, so that later for any given subsets $A' \subseteq A$, $B' \subseteq B$, $C' \subseteq C$, we can quickly decide if there exist $a \in A'$, $b \in B'$, $c \in C'$ such that $a+b=c$.

Chan and Lewenstein~\cite{ChanL15} showed how to answer such queries in $\Ot(n^{13/7})$ time\footnote{We use the standard $\Ot$ notation to suppress polylogarithmic factors.}. Their algorithm is based on the Balog--Szemer\'edi--Gowers (BSG) theorem from additive combinatorics, and is therefore quite complex.

More recently, Chan, Vassilevska Williams, and Xu~\cite{ChanWX23} proved what used to be a corollary of the BSG theorem in a more direct way, and as a result they obtained an improved running time of $\Ot(n^{11/6})$. Both algorithms~\cite{ChanL15,ChanWX23} use randomized\footnote{\cite{ChanL15} also gave a deterministic algorithm, but with preprocessing time increased to $\Ot(n^\omega)$, where $\omega$ denotes the matrix multiplication exponent.~\cite{ChanWX23} gave a deterministic algorithm too, but with query time increased to $O(n^{1.891})$, using bounds for rectangular matrix multiplication from~\cite{LeGallU18}.} preprocessing running in expected time $\Ot(n^2)$.

\subsection{Our results}
\label{sec:ourresults}

Our main result is a randomized $O(n^{1.5} \log n)$-time algorithm for 3SUM in preprocessed universes (with quadratic time preprocessing). Besides being faster than the previous algorithms, ours is also arguably simpler: we use only standard 3SUM-related tools, namely FFT and linear hashing modulo a prime. The algorithm is so simple that we actually sketch it in the following paragraph (we defer formal proofs to Section~\ref{sec:knownc}).

\paragraph{The algorithm.}
In the preprocessing phase, we pick a random prime $m \in [n^{1.5}, 2\cdot n^{1.5})$, and compute the set $F$ of all \emph{false positives} modulo $m$, i.e., triples $(a,b,c) \in A \times B \times C$ that are \emph{not} 3SUM solutions yet $a + b \equiv c \pmod{m}$. This can be done in time $O(n^2 + |F|)$, and by a standard argument $\E[|F|] \leqslant O(n^{1.5} \log n)$.

To answer a query, we first use FFT to compute, in time $O(m \log m)$, the number of 3SUM solutions modulo $m$, i.e., $\#\set{(a, b, c) \in A' \times B' \times C' \mid a + b \equiv c \pmod{m}}$. Then, we iterate over the list $F$ in order to compute, in time $O(|F|)$, the number of false positives present in the current instance, i.e., $|F \cap (A' \times B' \times C')|$. Finally, we subtract the two numbers, and answer ``yes'' if the result is nonzero.

\vskip 1em
\noindent
Just as the previous two algorithms~\cite{ChanL15,ChanWX23}, we can output in the same running time not only a single bit, indicating whether there is a 3SUM solution; rather, we can output $|C'|$ bits, indicating for each $c \in C'$ if there exist $a \in A'$, $b \in B'$ with $a+b=c$.\footnote{In literature this more general problem is sometimes called 3SUM\textsuperscript{+}~\cite{ChanL15} or All-Numbers 3SUM~\cite{ChanWX23}.} This is because both the number of solutions modulo $m$ and the number of false positives can easily be counted separately for each~$c \in C'$.

\paragraph{What if the set $C$ is unknown?}

Chan and Lewenstein~\cite{ChanL15} considered also a~variant of the problem where only the sets $A$ and $B$ are given for preprocessing, and hence the set $C'$ can contain any $n$ integers. For this variant they gave an algorithm with query time $\Ot(n^{19/10})$, i.e., a bit slower than in the known $C$ variant.  The improved algorithm of Chan, Vassilevska Williams, and Xu~\cite{ChanWX23} applies to the unknown~$C$ variant without any slowdown, i.e., the query time is~$\Ot(n^{11/6})$.

In Section~\ref{sec:unknowncrnd} we give an algorithm for that variant too. It has the same quadratic preprocessing time and $O(n^{1.5} \log n)$ query time as our algorithm for known $C$ from Section~\ref{sec:knownc}, and it is still simple. A notable difference is that the algorithm for known $C$ has space complexity $O(|F|) \leqslant O(n^{1.5} \log n)$ while the algorithm for unknown $C$ has space complexity $O(n^2)$.

\paragraph{Deterministic algorithms.}
Both our algorithms (for the variants with known $C$ and unknown $C$, presented in Sections~\ref{sec:knownc} and~\ref{sec:unknowncrnd}, respectively) use randomization in their simplest versions. However, both can be made deterministic with only polylogarithmic losses in the running times.

Our algorithm for known $C$ from Section~\ref{sec:knownc} can be derandomized by an essentially black-box application of a tool developed originally with the purpose of derandomizing fine-grained reductions from the 3SUM problem~\cite{FischerK024}. This tool is basically a~deterministic algorithm for selecting a (composite) modulus that gives a similar number of false positives as a randomly selected prime of the same magnitude. At the end of Section~\ref{sec:knownc} we comment on how it can be applied to derandomize our algorithm losing only a polylogarithmic factor in the running time.

In Section~\ref{sec:unknowncdet} we give a deterministic variant of our algorithm for unknown $C$ from Section~\ref{sec:unknowncrnd}, also slower by a polylogarithmic factor. Here, too, we use ideas from~\cite{FischerK024}, but we cannot apply them in a black-box manner, and we need an additional insight that while a single modulus might not suffice we can use $\log n$ moduli instead. This is the most technical part of our paper, but it is still relatively simple and only uses the same basic tools, namely FFT and bounds on the number of false positives.

See Table~\ref{tab:algorithms} for a comparison of the previous results with ours.

\renewcommand{\arraystretch}{1.3}
\begin{table}
\centering\small
\begin{tabular}{@{}llllcc@{}}
\toprule
Reference & Preprocessing & Query & Space & Unknown $C$ & Deterministic \\
\midrule
\cite{ChanL15} & $\Ot(n^2)$ & $\Ot(n^{13/7})$ & $O(n^{13/7})$ & -- & -- \\
\cite{ChanL15} & $\Ot(n^2)$ & $\Ot(n^{19/10})$ & $O(n^2)$ & \checkmark
 & -- \\
\cite{ChanL15} & $\Ot(n^\omega)$ & $\Ot(n^{13/7})$ & $O(n^{13/7})$ & -- & \checkmark
 \\
\cite{ChanL15} & $\Ot(n^\omega)$ & $\Ot(n^{19/10})$ & $O(n^2)$ & \checkmark & \checkmark
 \\
\cite{ChanWX23} & $\Ot(n^2)$ & $\Ot(n^{11/6})$ & $O(n^2 \log n)$ & \checkmark
 &  -- \\
\cite{ChanWX23} & $\Ot(n^2)$ & $O(n^{1.891})$ & $O(n^{1.891})$ & --
 & \checkmark
 \\
\textbf{Theorem~\ref{thm:knownc}} & $O(n^2)$ & $O(n^{1.5} \log n)$ & $O(n^{1.5} \log n)$ & -- & -- \\
\textbf{Theorem~\ref{thm:unknowncrnd}} & $O(n^2)$ & $O(n^{1.5} \log n)$ & $O(n^2)$ & \checkmark
 & --  \\
\textbf{Theorem~\ref{thm:unknowncdet}} & $O(n^2 \log n)$ & $O(n^{1.5} \log^3 n)$ & $O(n^2 \log n)$ & \checkmark
 & \checkmark \\
\bottomrule
\end{tabular}
\caption{Algorithms for 3SUM in preprocessed universes.}
\label{tab:algorithms}
\end{table}

\paragraph{Model of computation and input range.}

We work in the standard word RAM model with $O(\log n)$-bit words. It is common to assume that the input numbers fit a~single word, i.e., they are bounded by $n^{O(1)}$. We adopt this assumption, mainly for the sake of simplicity of presentation. We note that if the input numbers were instead bounded by an arbitrary universe size $U$ (but arithmetic operations were still assumed to take constant time), our algorithms would maintain essentially the same running time apart from all the $\log n$ factors replaced by $\log U$.

Let us also recall that, for the 3SUM problem specifically, one can assume using a by now standard argument (implicit in~\cite{BaranDP08}; formal proof, e.g., in~\cite[Lemma B.1]{AbboudLW14}) that the input numbers are bounded by $O(n^3)$, as long as randomization is allowed. This reduction applies also to the preprocessed universes variant, but using it would turn our Las Vegas algorithms into Monte Carlo ones and would likely involve an additional $\log n$ factor to guarantee the correct answer with high probability. For deterministic algorithms a similar reduction was given only recently~\cite{FischerK024}, but it does not seem to apply (at least directly) to the preprocessed universes variant, and it involves a $\operatorname{polylog} U$ factor anyway.

\subsection{Other related work}

Goldstein et al.~\cite{GoldsteinKLP17}, Golovnel et al.~\cite{GolovnevGHPV20}, and Kopelowitz and Porat~\cite{KopelowitzP19} studied a different type of preprocessing for 3SUM, called \emph{3SUM indexing}. In their variant, two sets $A$ and $B$ are given for preprocessing, and each query specifies a single integer $c$ and asks if there exist $a \in A$ and $b \in B$ with $a+b=c$. 
Hence, the two main aspects that differentiates this variant from the variant we study are that the query concerns (1) the original sets $A$ and $B$ (and not their subsets), and (2) just a single target $c$ (and not a set of up to $n$ targets $C'$). Also, the main focus of research on 3SUM indexing is the trade-off between the space complexity and the query time.

\paragraph{Reals.}

While in this paper we focus on the word RAM model, which is the most common choice in fine-grained complexity, we note that the 3SUM problem was originaly introduced in the real RAM model~\cite{GajentaanO95}, because its first applications were conditional lower bounds for problems in computational geometry, where the real RAM model prevails. It is conjectured that also in the real RAM there is no strongly subquadratic algorithm for 3SUM. The mildly subquadratic algorithm for 3SUM due to Baran, Demaine, and P{\u{a}}tra{\c{s}}cu~\cite{BaranDP08} works only in the word RAM model, but a separate line of research~\cite{GronlundP18,Chan20}, using very different techniques, led to similar polylogarithmic improvements also in the real RAM. Recently, Fischer~\cite{Fischer25} showed that the Chan--Lewenstein $\Ot(n^{13/7})$-time algorithm for 3SUM in preprocessed universes~\cite{ChanL15} can be adapted to work in the real RAM model. It seems to remain a difficult open question whether the subsequent improvement to time $\Ot(n^{11/6})$~\cite{ChanWX23} or our current improvement to time $O(n^{1.5} \log n)$ also can be mirrored in the real RAM.

\section{Preliminaries}

\subsection{Notation}

We use $[n] := \set{0, 1, \ldots, n-1}$, and $\Ot(T) := T \cdot (\log T)^{O(1)}$. For two sets of numbers $A$~and $B$ we define their \emph{sumset} to be $A + B := \set{a+b \mid (a, b) \in A \times B}$.

\subsection{Tools}
Throughout the paper we use the following folklore algorithm for sumset computation modulo a small number.

\begin{lemma}[Sumset computation using FFT]\label{lem:fft}
Given two sets $A$, $B$ of $n$ integers and a~positive integer $m$, one can compute in time $O(n + m \log m)$ the multiset $(A + B) \bmod m$, represented as its multiplicity vector $(v_0, v_1, \ldots, v_{m-1}) \in \mathbb{Z}^m$, where for every $i \in [m]$,
\[v_i := \#\set{(a, b) \in A \times B : a + b = i \tpmod{m}}.\]
\end{lemma}

\begin{proof}
Create two univariate polynomials $P := \sum_{a \in A}x^{a \bmod m}$ and $Q := \sum_{b \in B}x^{b \bmod m}$. Their degrees are at most $m - 1$. Then, use the Fast Fourier Transform (FFT) in order to compute the product polynomial $PQ$ in time $O(m \log m)$. Finally, for every $i \in [m]$, set $v_i$ to be the sum of coefficients in front of $x^i$ and $x^{m+i}$ in $PQ$.
\end{proof}

We also repeatedly use the following folklore fact about the number of false positives produced by hashing modulo a prime number.

\begin{lemma}[False positives modulo a prime]\label{lem:false}
Let $T \subseteq [U]^3$ be a set of triples of non-negative integers bounded by $U$ with no 3SUM solution, that is, $a + b \neq c$ for all $(a,b,c) \in T$. Let $r$ be an integer, and let $p$ be a prime number sampled uniformly at random from the range $[r, 2r)$. Then the expected number of false positives is bounded by
\[\E[\#\set{(a,b,c) \in T : a + b \equiv c \tpmod{p}}] \leqslant O\left(\frac{|T|\log U}{r}\right).\]
\end{lemma}

\begin{proof}
Let $P$ denote the set of primes in the range $[r,2r)$. By the prime number theorem there are $|P| = \Theta(r / \log r)$ of them. For a fixed $(a,b,c) \in T$, at most $\log_r(2U)$ of those primes divide $|a+b-c|$. This is because each of $a,b,c$ is between $0$ and $U-1$ (inclusive), so $-(U-1) \leqslant a+b-c \leqslant 2U - 2$, and in particular $|a+b-c| \leqslant 2U$. Therefore, the probability that $(a,b,c)$ becomes a false positive is bounded by
\[\operatornamewithlimits{Pr}_{p \sim P}[a + b \equiv c \tpmod{p}] \leqslant \frac{\log_r(2U)}{|P|} = \frac{\log 2U}{|P| \log r}.\]
By the linearity of expectation we get that the expected number of false positives is no greater than the above probability upper bound times $|T|$, i.e., $\frac{|T|\log 2U}{|P| \log r} \leqslant O(\frac{|T|\log U}{r})$.
\end{proof}
        
\section{Algorithm for known \texorpdfstring{$C$}{C}}
\label{sec:knownc}

In this section we present in detail our simplest algorithm, which we already sketch in Section~\ref{sec:ourresults}. The algorithm is randomized, and assumes that all three sets $A$, $B$, $C$ are available for preprocessing.
    
\begin{theorem}\label{thm:knownc}
    There is a randomized Las Vegas algorithm that,
    given three sets $A$, $B$, $C$ of $n$ integers in $[n^{O(1)}]$, can preprocess them in time $O(n^2)$ and space $O(n^{1.5} \log n)$, so that given any subsets $A' \subseteq A$, $B' \subseteq B$, $C' \subseteq C$, it can decide for all $c \in C'$ if there exist $a \in A'$, $b \in B'$ with $a+b=c$ in time $O(n^{1.5} \log n)$.
\end{theorem}

\paragraph{Preprocessing.}

The preprocessing boils down to picking a prime number $p$ uniformly at random from the range $[n^{1.5}, 2n^{1.5})$, and computing the set of \emph{false positives modulo~$p$}, which we define as
\[F:= \set{ (a,b,c) \in A \times B \times C : a + b \equiv c \tpmod{p} \land a + b \neq c }.\]

The expected number of false positives satisfies $\E[|F|] = O(\frac{n^3 \log n}{n^{1.5}}) = O(n^{1.5} \log n)$ as it follows from Lemma~\ref{lem:false} applied to the set of triples that are \emph{not} 3SUM solutions, i.e., $T=\set{(a,b,c) \in A \times B \times C : a + b \neq c}$, $|T| \leqslant n^3$.

Last but not least, let us argue that the set $F$ can be constructed in time $O(n^2+|F|)$. We create a length-$p$ array of (initially empty) lists, and for each number $c \in C$ we add it to the list stored at index $c \bmod p$ in the array. Then, for each of the $n^2$ pairs $(a,b) \in A \times B$, we go through the list at index $(a + b) \bmod p$; at most one element of this list is equal to $a+b$, and all the others constitute (together with $a$ and $b$) false positives, which we add to $F$.
Summarizing, the preprocessing takes time ${O}(n^2 + |F|)$, which is in expectation $O(n^2)$.

\paragraph{Query.}
First, we use FFT (see Lemma~\ref{lem:fft}) to compute the multiset $(A'+B') \bmod p$ in time $O(p \log p)$.
Then, we create a hash table $H$, which initially stores, for every $c \in C'$, the number of 3SUM solutions \emph{modulo $p$} involving $c$. We obtain this number in constant time from the output of FFT as the multiplicity of $c \bmod p$ in the multiset $(A'+B') \bmod p$. That is, we have
\begin{align*}
H[c] &= \#\{(a,b) \in A' \times B' : a + b \equiv c \tpmod{m}\} \\[.3em]
& = \#\smash{\underbrace{\{(a,b) \in A' \times B' : a+b=c\}}_{\text{true solutions}}} \\[.3em]
& \, + \hspace{12em}  \#\underbrace{\{(a,b) \in A' \times B' : a + b \equiv c \tpmod{m} \land a + b \neq c\}}_{\text{false positives}}.
\end{align*}

Our only remaining goal is to subtract the false positives so that we end up with the correct counts of solutions. To this end, we iterate over $F$. For each triple $(a,b,c) \in F$, we check whether $a \in A' \land b \in B' \land c \in C'$, and if it is the case then we decrement $H[c]$ by one.
Finally, for each $c \in C'$ we answer ``yes'' if $H[c] > 0$ and ``no'' otherwise.

Summarizing we answer the query in time $O(n + p \log p + |F|)$, which is $O(n^{1.5} \log n)$ in expectation.

\paragraph{Remarks.}
As presented above, the preprocessing and query time as well as space usage are bounded in expectation. However, we can slightly alter the preprocessing, so that as soon as we notice that the size of $F$ exceeds twice its expectation we immediately stop and start over with a freshly sampled prime. By Markov inequality this happens with probability at most $1/2$, so the expected number of such rounds is $1+1/2+1/4+\cdots=2$. Conveniently, now at the end of the preprocessing the set $F$ is \emph{guaranteed} to contain at most $O(n^{1.5} \log n)$ elements (not only in expectation), and the same guarantee applies therefore to the query time\footnote{The expectation in the query time also comes from the hash table, but we use it only for the sake of simplicity, and it can be avoided by indexing the elements of $A$, $B$, $C$ with integers from $1$ to $n$, which only requires additional time $O(n \log n)$.} and space usage.

A naive way to derandomize also the preprocessing (and therefore the whole algorithm) is to simply try out all primes in $P$ instead of picking a random one. That would increase the preprocessing time to $O(|P| \cdot n^2) = O(n^{3.5}/\log n)$, but the query time would remain unchanged.
A more efficient way is to use (a slight adaptation of) the deterministic algorithm of~\cite[Lemma~11]{FischerK024}, which can find in quadratic time a composite modulus $m$, being a product of three primes from the range $[\sqrt{n}, 2\sqrt{n})$ and generating only $O(n^{1.5} (\log n)^3)$ false positives. We do not elaborate on it further because the resulting deterministic algorithm for known $C$ would be subsumed (in all aspects but simplicity) by our deterministic algorithm for unknown $C$, which we present in Section~\ref{sec:unknowncdet}.    
    
\section{Randomized algorithm for unknown \texorpdfstring{$C$}{C}}
\label{sec:unknowncrnd}

In this section we show a modified variant of our algorithm from the previous section that does need to know the set $C$ at the time of preprocessing. The main challenge is that not knowing $C$ we cannot construct the set of false positives $F$. Instead, we prepare a simple data structure that allows us to efficiently enumerate the false positives once they become well defined during the query.

\begin{theorem}\label{thm:unknowncrnd}
    There is a randomized Las Vegas algorithm that,
    given two sets $A$, $B$ of $n$~integers in $[n^{O(1)}]$, preprocesses them in time $O(n^2)$, so that given any subsets $A' \subseteq A$, $B' \subseteq B$ and any set of $n$ integers $C'$, it can decide for all $c \in C'$ if there exist $a \in A'$, $b \in B'$ with $a+b=c$ in time $O(n^{1.5} \log n)$.
\end{theorem}

\paragraph{Preprocessing.}
We compute the sumset $A+B$, and for each $x \in A+B$ we store the set of \emph{witnesses of} $x$ defined as $W_x := \set{ (a, b) \in A \times B : a + b = x}$. We also pick a~random prime $m \in [n^{1.5}, 2\cdot n^{1.5})$. Finally, for every remainder $i \in [m]$ we prepare the list of sumset elements with that remainder, i.e., $L_i := \set{x \in A+B : x \equiv i \tpmod{m}}$. Clearly, the preprocessing takes $O(n^2)$ time.

\paragraph{Query.}
We use FFT to compute the multiset $(A'+B') \bmod m$ in time $O(m \log m) = O(n^{1.5} \log n)$. From that point on we handle each $c \in C'$ separately. First, we check if $c \in A + B \supseteq A' + B'$, and if not, we answer ``no''. Otherwise, we compute the number of 3SUM solutions involving $c$ using the following relation:
\begin{align*}
& \# \{(a,b) \in A' \times B' \mid a+b=c\} \\[.3em]
& = \# \smash{\underbrace{\{(a,b) \in A' \times B' \mid a + b \equiv c \tpmod{m}\}}_{\text{all solutions modulo $m$}}} \\[.3em]
&\, - \hspace{16em} \# \underbrace{\{(a,b) \in A' \times B' \mid a + b \equiv c \tpmod{m} \land a + b \neq c\}}_{\text{false positives}}.
\end{align*}

The first element of the subtraction is the multiplicity of $c \bmod m$ in the multiset $(A'+B') \bmod m$, which we read out in constant time from the output of FFT.

For the second element of the subtraction, we iterate over $x \in L    _{c \bmod m}$ and for each $x \neq c$ we go over $(a,b) \in W_x$ and add one when $a \in A'$ and $b \in B'$. In other words, we count the false positives using the following identity
\[
\{(a,b) \in A' \times B' \mid a + b \equiv c \tpmod{m} \land a + b \neq c\}
= \bigcup_{\mathclap{\substack{x \in L_{c \bmod m} \\ x \neq c}}} \: (W_x \cap (A' \times B')).
\]
This step takes time proportional to the total size of sets of witnesses that we go over, that is, $\sum_{x \in L_{c \bmod m}, x \neq c} |W_x|$, which is exactly
$\#\set{(a, b) \in A \times B : a+b \equiv c \tpmod{m} \land a+b\neq c}$, i.e., the number of false positives involving $c$ among $A$ and $B$ (and it is potentially different than the number of false positives among $A'$ and $B'$, which we calculate). By Lemma~\ref{lem:false} applied to $T = \set{(a,b,c) \in A \times B \times \set{c} : a + b \neq c}$ this number is in expectation $O(\sqrt{n} \log n)$. In total, we process the query in expected time $O(m \log m + |C'|\sqrt{n}\log n) = O(n^{1.5}\log n)$.

\paragraph{Remarks.}
Contrary to the algorithm from Section~\ref{sec:knownc} for known $C$, the query time upper bound of the above algorithm seems inherently probabilistic. We do not see any easy way to restrict all the uncertainty about the running time to the preprocessing phase (without resorting to some of the derandomization ideas from Section~\ref{sec:unknowncdet}). Moreover, the space complexity of this algorithm is $O(n^2)$, as we have to store the entire sumset $A+B$.

\section{Deterministic algorithm for unknown \texorpdfstring{$C$}{C}}
\label{sec:unknowncdet}

In this section we derandomize our algorithm for unknown $C$ from the previous section. Unlike our algorithm for known $C$ (from Section~\ref{sec:knownc}), which can be easily derandomized by a black-box application of a known tool, this one requires more work. The challenge comes from the fact that in the preprocessing, not knowing $C$, we cannot simply pick a modulus that has few false positives for $C$ (and hence also for $C' \subseteq C$). Fortunately, we do know another superset of (the nontrivial part\footnote{For any $c\in C'\setminus (A+B)$ we can immediately answer ``no''.} of) $C'$, namely $A+B$. This set however can be much larger than $n$, and as a result we are not guaranteed to find a single modulus that works well for the whole set. We overcome this last issue by finding $\log n$ moduli such that for each potential element of $C'$ at least one of them is good enough.

\begin{theorem}\label{thm:unknowncdet}
    There is a deterministic algorithm that, given two sets $A$, $B$ of $n$ integers in $[n^{O(1)}]$, preprocesses them in time $O(n^2 \log n)$, so that given any subsets $A' \subseteq A$, $B' \subseteq B$ and any set of $n$ integers $C'$, it can decide for all $c \in C'$ if there exist $a \in A'$, $b \in B'$ with $a+b=c$ in time $O(n^{1.5} \log^3 n)$.
\end{theorem}

\paragraph{Preprocessing.}

First, we compute the sumset $A+B$ by naively iterating over all pairs $(a,b) \in A \times B$, and for each $x \in A+B$ we store the set of \emph{witnesses of} $x$ defined as
\[W_x := \set{ (a, b) \in A \times B : a + b = x}.\]
Not being able to use a (randomized) hash table, we keep $W_x$'s in a dictionary (indexed by $x$) implemented using a binary search tree, so this first part of the preprocessing already takes time $O(n^2 \log n)$.

We say that an element $x \in A+B$ is \emph{light} if $|W_x| \leqslant \sqrt{n}$, and otherwise we say that $x$ is \emph{heavy}. There are at most $n^{1.5}$ heavy elements.
This is because the sets of witnesses are disjoint, and their union is $A \times B$, so their total size is $n^2$, and hence at most $n^2/\sqrt{n}$ of them can be of size $\sqrt{n}$ or larger.

Next we find a set $M \subseteq \mathbb{Z}$ of moduli in the range $[n^{1.5}, 8 \cdot n^{1.5})$, of size $|M|=O(\log n)$, such that for every heavy $x$ there exists a modulus $m \in M$ such that
\[\#\set{(a,b) \in A \times B : a + b \equiv x \tpmod{m} \land a+b \neq x} \leqslant O(\sqrt{n} \log^3 n). \tag{$\star$}\]

Finally, for every modulus $m \in M$ and every remainder $i \in [m]$ we prepare the list of sumset elements with that remainder, i.e., $L^m_i := \set{x \in A+B : x \equiv i \tpmod{m}}$. Note that these lists contain both heavy and light elements.

Before we proceed to argue that such a set $M$ at all exists and that it can be found efficiently, let us show how it can be used to answer queries in time $O(n^{1.5} \log^3 n)$.

\paragraph{Query.}
First, we use FFT to compute the multiset $(A'+B') \bmod m$, for every $m \in M$. This takes $O(n^{1.5} \log^2 n)$ time in total. From that point on we handle each $c \in C'$ separately. If $c \notin A+B$, we answer ``no'' immediately. If $c$ is light, we just need to check if at least one of its at most $\sqrt{n}$ witnesses is present in the current instance, i.e., whether $W_x \cap (A' \times B') \neq \emptyset$, which takes time $O(|W_x|) \leqslant O(\sqrt{n})$. Finally, if $c$ is heavy, we pick a modulus $m \in M$ that satisfies ($\star$) for that $c$, and we compute the number of 3SUM solutions involving $c$ using the following relation:
\begin{align*}
& \# \{(a,b) \in A' \times B' \mid a+b=c\} \\[.3em]
& = \# \smash{\underbrace{\{(a,b) \in A' \times B' \mid a + b \equiv c \tpmod{m}\}}_{\text{all solutions modulo $m$}}} \\[.3em]
&\, - \hspace{16em} \# \underbrace{\{(a,b) \in A' \times B' \mid a + b \equiv c \tpmod{m} \land a + b \neq c\}}_{\text{false positives}}.
\end{align*}

The first element of this subtraction equals the multiplicity of $c \bmod m$ in the multiset $(A'+B') \bmod m$, which we read out in constant time from the output of FFT. To get the second element of the subtraction, we iterate over $x \in L^m_{c \bmod m}$ and for each $x \neq c$ we go over $(a,b) \in W_x$ and count one when $a \in A'$ and $b \in B'$. In other words, we express the set of false positives as 
$\bigcup\set{W_x \cap A' \times B' : x \in L^m_{c \bmod m} \land x \neq c}$. This works in time $O\big(\sum_{x \in L^m_{c \bmod m} \land x \neq c} |W_x|\big)$, which is bounded by $O(\sqrt{n} \log^3 n)$ thanks to ($\star$). Summarizing, for each $c \in C'$ we spend at most $O(\sqrt{n} \log^3 n)$ time, so the total running time is $O(n^{1.5} \log^3 n)$.

\paragraph{Finding a set of moduli.}
In the preprocessing, in order to find $M$, we first initialize $X$ to be the set of all the heavy elements of $A+B$, and $M$ to be the empty set. Then, as long as $X$ is non-empty, we (1) find a modulus $m$ such that the average number of false positives per element of $X$ is $O(\sqrt{n} \log^3 n)$, (2) add $m$ to $M$, and (3) remove from $X$ those elements for which the number of false positives modulo $m$ is at most twice the average. There are at least $|X|/2$ such elements removed, so the process stops after $\log |X| = O(\log n)$ moduli are added to $M$. In order to keep the total preprocessing time $O(n^2 \log n)$, we need to find one such modulus $m$ in time $O(n^2)$. To this end we adapt~\cite[Lemma 11]{FischerK024}.

\begin{lemma}[Adapted from \cite{FischerK024}]
\label{lem:modulusselection}
Given three sets of integers $A, B, X \subseteq [U]$, with $|A|=|B|=n$, there exists a modulus $m \in [n^{1.5}, 8\cdot n^{1.5})$ such that the average number of false positives per element of $X$ satisfies the following upper bound
\[\frac{\#\set{(a,b,x) \in A \times B \times X : a + b \equiv x \tpmod{m} \land a + b \neq x}}{|X|} \leqslant O(\sqrt{n} \cdot (\log U)^3).\]
Moreover, there is a deterministic algorithm that finds such $m$ in time $O(n^2 + |X|\frac{\sqrt{n}}{\log n})$.
\end{lemma}

\begin{proof}
For a positive integer $m$, let $\mu(m)$ denote the number of 3SUM solutions modulo $m$, i.e.,
$\mu(m) := \#\set{(a,b,x) \in A \times B \times X : a + b \equiv x \tpmod{m}}$. 
Note that $\mu(m)$ can be computed in $O(n + m \log m + |X|)$, by first computing the multiset $(A+B) \bmod m$ in time $O(n + m \log m)$ using FFT (Lemma~\ref{lem:fft}), and then summing over $x \in X$ the multiplicity of $x \bmod m$ in that multiset.
Our goal is to effectively construct $m$ such that
\[\mu(m) \leqslant |\set{(a,b,x) \in A \times B \times X : a + b = x}| + |X| \cdot O(\sqrt{n} (\log U)^3).\]

Let us first describe the algorithm. After that we analyse it, proving along the way that such a modulus $m$ indeed exists. Let $P$ be the set of all primes in the range $[\sqrt{n}, 2\sqrt{n})$, and let $m_0 = 1$. For $i=1,2,3$ do the following: for every $p \in P$ compute $\mu(m_{i-1}p)$, pick $p$ that minimizes that number, and set $m_i = m_{i-1}p$. After these three iterations, return $m_3$.

The running time of the algorithm is simple to analyse. In each of the three iterations we examine $|P| = O(\sqrt{n} / \log n)$ different primes, and for each of them we compute the number of solutions modulo a number less than $(2\sqrt{n})^3$, so the running time is 
\[O(|P|\cdot(n + n^{1.5}\log n + |X|) = O\big(n^2+|X|\frac{\sqrt{n}}{\log n}\big).\]

Let us now argue that $m_3$ actually satisfies the guarantees of the lemma.
First, $m_3$~is a product of three numbers in the range $[\sqrt{n}, 2\sqrt{n})$, so it is clearly in the desired range $[n^{1.5}, 8 \cdot n^{1.5})$.
Second, for a positive integer $m$, let $F(m)$ denote the set of false positives modulo $m$, i.e.,
\[F(m) := \set{(a,b,x) \in A \times B \times X : a + b \equiv x \tpmod{m} \land a + b \neq x}.\]
It remains to show that $|F(m_3)| \leqslant O(|X|\sqrt{n}\cdot(\log U)^3)$.
To this end, we claim that $|F(m_i)| \leqslant O(|F(m_{i-1})| \log U / \sqrt{n})$ for every $i\in\set{1,2,3}$. Indeed, we have
\[ \min_{p \in P} |F(m_{i-1}p)|  \leqslant \E_{p \in P}[|F(m_{i-1}p)|] \leqslant O\left(\frac{|F(m_{i-1})| \log U}{\sqrt{n}}\right),\]
where the second inequality follows from Lemma~\ref{lem:false} applied to $F(m_{i-1})$. The algorithm picks $p$ minimizing $\mu(m_{i-1}p) = |\set{(a,b,x) \in A \times B \times X : a + b = x}| + |F(m_{i-1}p|$, so this $p$ also minimizes $|F(m_{i-1}p)|$, and the claim follows.

To finish the proof note that $|F(m_0)| \leqslant |A \times B \times X| = |X| \cdot n^2$, and applying the above claim three times we obtain
\[|F(m_3)| \leqslant O\left(|F(m_0)| \cdot \left(\frac{\log U}{\sqrt n}\right)^3\right) = O(|X| \sqrt{n} \cdot (\log U)^3).\qedhere\]
\end{proof}

\printbibliography

\end{document}